\documentclass[12pt,twoside]{article}

\usepackage{lineno,hyperref,mathtools}
\usepackage{float}
\usepackage{threeparttable}
\usepackage{amssymb}
\usepackage{bm}
\usepackage{amssymb}
\usepackage{amsthm}
\usepackage{multirow,booktabs}
\usepackage{cases}
\usepackage{tabularx}
\usepackage{adjustbox}
\usepackage[figuresright]{rotating}
\usepackage{longtable}
\usepackage{booktabs}
\usepackage{color}
\usepackage{setspace}
\usepackage{float}
\usepackage{graphicx}
\usepackage{epstopdf}
\usepackage{graphicx}
\usepackage{epic}
\usepackage{multirow}
\usepackage{tikz}
\usepackage{xcolor}
\usepackage{makecell}

\usetikzlibrary{arrows,shapes,chains}

\renewcommand{\paragraph}{\roman{paragraph}}
\usepackage[a4paper]{geometry}
\makeatletter
\renewcommand\title[1]{\gdef\@title{\reset@font\Large\bfseries #1}}
\renewcommand\section{\@startsection {section}{1}{\z@}%
                                   {-3.5ex \@plus -1ex \@minus -.2ex}%
                                   {2.3ex \@plus.2ex}%
                                   {\normalfont\large\bfseries}}
\renewcommand\subsection{\@startsection{subsection}{2}{\z@}%
                                     {-3ex\@plus -1ex \@minus -.2ex}%
                                     {1.5ex \@plus .2ex}%
                                     {\normalfont\normalsize\bfseries}}
\renewcommand\subsubsection{\@startsection{subsubsection}{3}{\z@}%
                                     {-2.5ex\@plus -1ex \@minus -.2ex}%
                                     {1.5ex \@plus .2ex}%
                                     {\normalfont\normalsize\bfseries}}

\def\@runningauthor{}\newcommand{\runningauthor}[1]{\def\runningauthor{#1}}
\def\@runningtitle{}\newcommand{\runningtitle}[1]{\def\runningtitle{#1}}

\renewcommand{\ps@plain}{%
\renewcommand{\@evenhead}{\footnotesize\scshape \hfill\runningauthor\hfill}
\renewcommand{\@oddhead}{\footnotesize\scshape \hfill\runningtitle\hfill}}

\newcommand{\F}{\mathbb{F}}

\newcommand {\ccc}{{\mathbf{c}}}
\newcommand {\C}{{\mathcal{C}}}
\newcommand {\PP}{{\mathcal{P}}}

\newcommand{\GRM}{{\mathrm{GRM}}}

\newcommand{\bc}{{{\bf c}}}
\newcommand{\bx}{{{\bf x}}}
\newcommand{\by}{{{\bf y}}}

\newcommand{\uuu}{{{\mathbf{u}}}}
\newcommand{\vvv}{{{\mathbf{v}}}}
\newcommand{\wt}{{{\mathbf{wt}}}}

\g@addto@macro\bfseries{\boldmath}

\makeatother

\usepackage{amsthm,amsmath,amssymb}
\usepackage{cite}
\usepackage{graphicx}

\theoremstyle{plain}
\newtheorem{theorem}{Theorem}[section]

\newtheorem{lemma}[theorem]{Lemma}
\newtheorem{corollary}[theorem]{Corollary}

\theoremstyle{definition}
\newtheorem{definition}[theorem]{Definition}
\newtheorem{example}[theorem]{Example}

\theoremstyle{remark}
\newtheorem{remark}[theorem]{Remark}

\runningauthor{}

\date{}

\begin{document}
\begin{sloppypar}

\title{Symplectic self-orthogonal and LCD codes from the Plotkin sum construction}
\author{Shixin Zhu, Yang Li\thanks{Corresponding author}, Shitao Li
\thanks{This research is supported by the National Natural Science Foundation of China (U21A20428, 12171134 and 12001002).}
\thanks{Shixin Zhu and Yang Li are with the School of Mathematics, Hefei University of Technology, Hefei, 230601, China 
(email: zhushixinmath@hfut.edu.cn, yanglimath@163.com). 
Shitao Li is with the School of Mathematical Sciences, Anhui University, Hefei, 230601, China, (email: lishitao0216@163.com).}}

    \maketitle

  \begin{abstract}
      In this work, we propose two criteria for linear codes obtained from the Plotkin sum construction being symplectic self-orthogonal (SO) and 
      linear complementary dual (LCD). 
      As specific constructions, several classes of symplectic SO codes with good parameters including symplectic maximum distance separable codes 
      are derived via $\ell$-intersection pairs of linear codes and generalized Reed-Muller codes. 
      Also symplectic LCD codes are constructed from general linear codes. 
      Furthermore, we obtain some binary symplectic LCD codes, which are equivalent to quaternary trace Hermitian additive 
      complementary dual codes that outperform best-known quaternary Hermitian LCD codes 
      reported in the literature. 
      In addition, we prove that symplectic SO and LCD codes obtained in these ways are asymptotically good. 
  \end{abstract}


\section{Introduction}\label{sec-introduction}

Throughout this paper, let $q=p^m$ be a prime power and $\F_q$ be the finite field with size $q$. 
An $[n,k]_q$ \emph{linear code} $\C$ is a $k$-dimensional linear subspace of $\F_q^n$. 
Let $\C^{\perp}$ be the dual code of $\C$ with respect to a certain inner product (such as the Euclidean, Hermitian or symplectic inner product). 
A linear code $\C$ is said to be \emph{self-orthogonal (SO)} if $\C\subseteq \C^{\perp}$ 
and \emph{linear complementary dual (LCD)} if $\C\cap \C^{\perp}=\{\mathbf{0}\}$.  
Both SO and LCD codes have attracted significant attention in recent years due to their theoretical and practical importance.

On one hand, constructing, enumerating, characterizing and classifying SO codes have remained as four essential 
and dynamic research problems since the beginning of coding theory 
(see \cite{BBGO2006,CPS1992,KC2022,KKL2021,P1972,KL2004,FL2022,Xu,LZM2023} and the references therein).  
Two main factors contribute to the intriguing and broad appeal of SO codes. 
First, Ding \cite{D2009SOgood1}, Zhang $et\ al.$ \cite{ZC2022SOgood2} and Jin $et\ al.$ \cite{JX2011} 
respectively proved that binary Euclidean SO codes, $q$-ary Euclidean SO codes ($q$ is odd) and 
$q$-ary symplectic SO codes are asymptotically good. 
Second, extensive researches have established strong correlations between SO codes and various mathematical fields 
such as combinatorial $t$-design theory \cite{t-design}, group theory \cite{C-1}, lattice theory \cite{C-1,B-1,H-1}, 
modular forms \cite{Shi-book}, and quantum error-correcting codes (QECCs) \cite{CRSS1998,quantum-codes-IT-2}. 
Specifically, finite groups like the Mathieu groups were found to be associated with some SO codes and the extended binary 
SO Golay code was linked to the Conway group. 
Many $5$-designs were also obtained from SO codes \cite{5-design}. 


On the other hand, LCD codes were first introduced by Massey \cite{M1992} in 1992, 
which can provide an optimum linear coding solution for two-user 
binary adder channel. Subsequently, Sendrier \cite{LCD-is-good} and G\"uneri 
$et\ al.$ \cite{GOS2016} showed that Euclidean and Hermitian LCD codes are asymptotically 
good. Carlet $et\ al.$ \cite{CG2016SCA} also further developed LCD codes to combat 
side channel attacks (SCAs) and fault injection attacks (FIAs). In 2018, Carlet 
$et\ al.$ \cite{CG2016} proved that any linear code over $\F_q$ is equivalent 
to some Euclidean LCD code for $q>3$ and any linear code over $\F_{q^2}$ is 
equivalent to some Hermitian LCD code for $q>2$. Since then, the focus has been on 
binary, ternary Euclidean LCD codes, quaternary Hermitian LCD codes and $q$-ary 
symplectic LCD codes 
(see for example \cite{AH2020,AHS2020,AHS2021.1,GKLRW2018,234-lcD,HS-BLCD-1-16,LLY2022,XD2021SLCD,HLH2022,LZM2023}).  
Note that Xu $et\ al.$ \cite{XD2021SLCD} and Huang $et\ al.$ \cite{HLH2022} also 
employed symplectic LCD symplectic maximum distance separable (MDS) codes  
to construct maximal entanglement MDS entanglement-assisted QECCs.

Hence, it is always interesting to construct new SO and LCD codes with good parameters and  
it should also be emphasized that a variety of effective techniques have been developed in the literature. 
In particular, one of such excellent methods is the 
so-called Plotkin sum construction \cite{Huffman}, also referred to as the $(\mathbf{u},\mathbf{u+v})$ construction, 
which can generate new linear codes 
from old ones. Very recently, by employing the Plotkin sum construction, 
Li $et\ al.$ \cite{LZM2023} constructed many good Euclidean and Hermitian SO and LCD codes.  
Motivated by this work and the growing interest in SO and LCD codes, a natural problem 
arises: \textbf{Can symplectic SO and LCD codes with good parameters be constructed from 
the Plotkin sum construction?} In this paper, we provide an affirmative answer. 
Main contributions of ours are summarized as follows:

\begin{enumerate}
    \item [\rm (1)] Criteria for linear codes obtained from the Plotkin sum construction 
    being symplectic SO and symplectic LCD codes are characterized in Theorems 
    \ref{th.Plotkin sso} and \ref{th.Plotkin SLCD}.  
    Also symplectic SO and LCD codes obtained from these ways are proved 
    to be asymptotically good with respect to the symplectic distance 
    in Theorems \ref{th.Plotkin sum SO good} and \ref{th.Plotkin sum LCD good}. 
    
    \item [\rm (2)] Using these two criteria, we further obtain many symplectic SO and LCD codes with good parameters. 
\begin{itemize}
    \item By utilizing $\ell$-intersection pairs of linear codes 
    and generalized Reed-Muller codes, we construct several families of symplectic SO 
    codes with explicit parameters in Theorems \ref{th.symplectic SO.l-intersection}, 
    \ref{th.symplectic SO and DC MDS codes}, \ref{th.symplectic self-dual} 
    and \ref{th.SSO from GRM} as well as Corollary \ref{coro.symplectic self-dual codes}. 
    Many symplectic dual-containing (DC) codes are also produced by considering 
    symplectic dual codes of these symplectic SO codes. 
    As results, lots of symplectic SO and DC codes with good parameters including 
    symplectic MDS codes are deduced.   

    \item By employing general linear codes, 
    we present a method to construct symplectic LCD codes 
    in Theorem \ref{th.2-ary symplectic LCD codes}. 
    Based on this method, we further obtain some binary symplectic LCD codes, 
    which are equivalent to the so-called quaternary trace Hermitian additive complementary dual codes 
    that outperform best-known quaternary Hermitian LCD codes in Table 2. 
\end{itemize} 
\end{enumerate}

The paper is organized as follows. 
In Section \ref{sec2}, we review some necessary knowledge. 
In Sections \ref{sec3} and \ref{sec4}, we apply the Plotkin sum construction 
to obtain symplectic SO, DC and LCD codes with explicit and good parameters. 
We also study the asymptotic results of symplectic SO and LCD codes from 
the Plotkin sum construction. 
In Section \ref{sec5}, we conclude this paper.

\section{Preliminaries}\label{sec2}

\subsection{Linear codes}\label{sec2.1}

From now on, we always denote $\mathbf{0}$ as an appropriate zero vector and $O$ as a proper zero matrix. 
Let $\bx_1=(x_{1},x_{2},\ldots,x_{n})$, $\bx_{2}=(x_{n+1},x_{n+2},\ldots,x_{2n})$, 
$\by_1=(y_{1},y_{2},\ldots,y_{n})$ and $\by_{2}=(y_{n+1},y_{n+2},\ldots,y_{2n})$ be any four vectors in $\F_q^{n}$. 
The \emph{Hamming weight} of $\bx_1$ is $\wt_{\rm H}(\bx_1)=|\{i\mid x_i\neq 0, 1\leq i\leq n\}|$ and 
the \emph{minimum Hamming distance} of an $[n,k]_q$ linear code $\C$ is 
$d_{\rm H}(\C)=\min\{\wt_{\rm H}(\bx_1)\mid \bx_1\in \C\ {\rm and}\ \bx_1\neq \mathbf{0}\}$. 
Denote $\bx=(\bx_1\mid \bx_2)=(x_{1},\ldots,x_{n},x_{n+1},\ldots,x_{2n})$ and 
$\by=(\by_1\mid \by_2)=(y_{1},\ldots,y_{n},y_{n+1},\ldots,y_{2n})$. 
The \emph{symplectic weight} of $\bx$ is $\wt_{\rm s}(\bx)=|\{i\mid (x_i,x_{n+i})\neq (0,0), 1\leq i\leq n\}|$ and 
the \emph{minimum symplectic distance} of a $[2n,k]_q$ linear code $\C$ is 
$d_{\rm s}(\C)=\min\{\wt_{\rm s}(\bx)\mid \bx\in \C\ {\rm and}\ \bx\neq \mathbf{0}\}$. 
In this paper, we denote $\C$ by $[n,k,d_{\rm H}]_q^{\rm H}$ (resp. $[2n,k,d_{\rm s}]_q^{\rm s}$) 
if $\C$ is an $[n,k]_q$ (resp. a $[2n,k]_q$) linear code with minimum Hamming (resp. symplectic) distance $d_{\rm H}$ (resp. $d_{\rm s}$).  
It is well-known that for an $[n,k,d_{\rm H}]_q^{\rm H}$ (resp. a $[2n,k,d_{\rm s}]_q^{\rm s}$) linear code $\C$, 
the \emph{Hamming (resp. symplectic) Singleton bound} says that 
$d_{\rm H}\leq n-k+1$ (resp. $d_{\rm s}\leq \lfloor \frac{2n-k+2}{2} \rfloor$). 
Then such a linear code $\C$ is called a \emph{Hamming (resp. symplectic) MDS code} 
if $d_{\rm H}=n-k+1$ (resp. $d_{\rm s}=\lfloor \frac{2n-k+2}{2} \rfloor$). 

Let $\Omega=\left[\begin{array}{cc}
    O & I_n \\ 
    -I_n & O
\end{array}\right]$, where $I_n$ is the identity matrix of size $n\times n$.  
The \emph{Euclidean inner product} of $\bx_1$ and $\by_1$ is defined by 
\begin{align}
    \langle \bx_1, \by_1 \rangle_{\rm E}=\sum_{i=1}^{n}x_{i}y_{i}
\end{align}
and the \emph{symplectic inner product} of $\bx$ and $\by$ is defined by 
\begin{align}
    \langle \bx, \by \rangle_{\rm s}= \bx \Omega \by^T = \sum_{i=1}^{n}(x_iy_{n+i}-x_{n+i}y_i).  
\end{align}
If $\C$ is an $[n,k,d_{\rm H}]_q^{\rm H}$ linear code, its \emph{Euclidean dual code} is given by 
\begin{align}\label{eq.Edual}
    \C^{\perp_{\rm E}}=\{\by_1\mid \langle \bx_1, \by_1 \rangle_{\rm E}=0,\ \forall\ \bx_1\in \C\}  
\end{align}
and $\C^{\perp_{\rm E}}$ has parameters $[n,n-k,d_{\rm H}^{\perp_{\rm E}}]_q^{\rm H}$, where $d_{\rm H}^{\perp_{\rm E}}$ denotes the minimum Hamming distance of $\C^{\perp_{\rm E}}$.  
If $\C$ is a $[2n,k,d_{\rm s}]_q^{\rm s}$ linear code, its \emph{symplectic dual code} is given by 
\begin{align}\label{eq.sdual}
    \C^{\perp_{\rm s}}=\{\by\mid \langle \bx, \by \rangle_{\rm s}=0,\ \forall\ \bx\in \C\}  
\end{align}
and $\C^{\perp_{\rm s}}$ has parameters $[n,n-k,d_{\rm s}^{\perp_{\rm s}}]_q^{\rm s}$, where $d_{\rm s}^{\perp_{\rm s}}$ denotes the minimum symplectic distance of $\C^{\perp_{\rm s}}$. 
As defined previously, we call $\C$ \emph{Euclidean} (resp. \emph{symplectic}) \emph{SO} if $\C\subseteq \C^{\perp_{\rm E}}$ (resp. $\C\subseteq \C^{\perp_{\rm s}}$), 
call $\C$ \emph{Euclidean} (resp. \emph{symplectic}) \emph{DC} if $\C^{\perp_{\rm E}}\subseteq \C$ (resp. $\C^{\perp_{\rm s}}\subseteq \C$), and 
call $\C$ \emph{Euclidean} (resp. \emph{symplectic}) \emph{LCD} if $\C\cap \C^{\perp_{\rm E}}=\{\mathbf{0}\}$ (resp. $\C\cap \C^{\perp_{\rm s}}=\{\mathbf{0}\}$). 
In addition, it is easy to check that $(\C^{\perp_{\rm E}})^{\perp_{\rm E}}=\C$ and $(\C^{\perp_{\rm s}})^{\perp_{\rm s}}=\C$. 
Note that symplectic SO and LCD codes are characterized by the following two lemmas. 
\begin{lemma}[Theorem 1 in \cite{Xu}]\label{lem.symplectic SO}
    Let $\C$ be a $[2n,k]_q$ linear code with a generator matrix $G$. 
    Then $\C$ is a symplectic SO code if and only if $G\Omega G^T$ is a zero matrix.   
\end{lemma}

\begin{lemma}[Theorem 1 in \cite{XD2021SLCD}]\label{lem.symplectic LCD}
    Let $\C$ be a $[2n,k]_q$ linear code with a generator matrix $G$. 
    Then $\C$ is a symplectic LCD code if and only if $G\Omega G^T$ is nonsingular.    
\end{lemma}

\subsection{Additive codes}\label{sec2.2}

An $(n,q^k)_{q^2}$ {\em additive code} is an $\F_q$-linear subgroup of $\F_{q^2}^n$, which has size $q^k$ for some integer $k$ satisfying $0\leq k\leq 2n$. 
An $(n,q^k, d_{\rm H})_{q^2}^{\rm H}$ {\em additive code} is an $(n,q^k)_{q^2}$ additive code with minimum Hamming distance $ d_{\rm H}$.
Let $\F_{q^2}^*=\langle\omega\rangle$.
For ${\bf u} = (u_1, u_2, \ldots, u_n)$ and $\textbf v = (v_1, v_2, \ldots, v_n)\in \F_{q^2}^{n}$,
the {\em alternating form} of ${\bf u}$ and ${\bf v}$ is defined by
\begin{align}\label{inner-alter}
  \langle \textbf u, \textbf v\rangle_{\rm a} =\sum_{i=1}^n \frac{u_iv_i^q- u_i^qv_i}{\omega-\omega^q}.
\end{align}
We use $\mathcal{C}^{\perp_{\rm a} }$ to denote the dual code of an additive code $\mathcal{C}$ under the alternating form  
and call $\mathcal{C}$ {\em additive SO} if $\mathcal{C} \subseteq \mathcal{C}^{\perp_{\rm a} }$ 
and {\em additive complementary dual (ACD)} if $\mathcal{C} \cap \mathcal{C}^{\perp_{\rm a} }=\{{\bf 0}\}$. 
Consider the following map
\begin{align*}
  \phi: ~\F_{q}^{2n} &\rightarrow \F_{q^2}^n, \\
  (a_1,\ldots,a_n,a_{n+1},\ldots,a_{2n}) &\mapsto (a_1+\omega a_{n+1},\ldots,a_{n}+\omega a_{2n}).
\end{align*}
It can be checked that $\langle {\bf u},{\bf v}\rangle_{\rm s} =\langle \phi({\bf u}), \phi({\bf v})\rangle_{\rm a} $ 
for any ${\bf u},{\bf v}\in \F_q^{2n}$ (see \cite{quantum-codes-IT-2}).
It is also easy to verify that $\phi$ is an isomorphic map from ($\F_q^{2n}$, $d_{\rm s}$, $\langle \cdot,\cdot\rangle_{\rm s} $) to ($\F_{q^2}^n$, $d_{\rm H}$, $\langle \cdot,\cdot\rangle_{\rm a} $).
Therefore, a $[2n,k,d]_q^{\rm s}$ linear code is equivalent to an $(n,q^{k},d)^{\rm H}_{q^2}$ additive code.

\subsection{Asymptotic results}\label{sec2.3}

Finally, we end this section with the concept of asymptotic results. 
For an infinite subset of  $[n_i,k_i,d_i]_q^{\rm s}$ codes from a family of linear codes with $\lim_{i \rightarrow \infty}n_i=\infty$, 
if both of its \emph{asymptotic rate}  $$r=\liminf_{i \rightarrow \infty}\frac{k_i}{n_i}$$ 
and \emph{asymptotic relative Hamming (resp. symplectic) distance} $$\delta=\liminf_{i \rightarrow \infty}\frac{d_i}{n_i}$$ 
are greater than $0$, we call this family of linear code \emph{asymptotically good with respect to the Hamming (resp. symplectic) distance}. 
As pointed out in \cite{Huffman}, 
asymptotically good linear codes are generally considered interesting and one prefers to delve into a family of asymptotically good codes 
in coding theory.

\section{Symplectic SO codes}\label{sec3}
\subsection{The characterization and the asymptotic result of symplectic SO codes from the Plotkin sum construction}\label{sec3.1}

\begin{definition}[\cite{Huffman}]\label{def.Plotkin sum}
    Let $\C_i$ be an $[n,k_i]_q$ linear code for $i=1, 2$. 
    Let $G_i$ and $H_i$ be respectively a generator matrix 
    and a parity check matrix of $\C_i$ for $i=1, 2$. 
    The \emph{Plotkin sum} of $\C_1$ and $\C_2$ is a 
    $[2n,k_1+k_2]_q$ linear code $\PP(\C_1,\C_2)$ defined by 
    \begin{align}\label{eq.def.Plotkin sum}
        \PP(\C_1,\C_2)=\{(\mathbf{u},\mathbf{u+v})\mid \mathbf{u}\in \C_1,\ \mathbf{v}\in \C_2\}, 
    \end{align}
    whose generator matrix and parity check matrix are respectively 
    \begin{align}\label{eq.Plotkin sum matrix}
        G=\left(\begin{array}{cc}
            G_1 & G_1 \\
            O & G_2 
        \end{array}\right)\ {\rm and}\  H=\left(\begin{array}{cc}
            H_1 & O \\
            -H_2 & H_2 
        \end{array}\right).
    \end{align}
\end{definition}


\begin{theorem}\label{th.Plotkin sso}
	Let $\C_i$ be an  $[n,k_i,d_{i}]_q^{\rm H}$ linear code for $i=1, 2$.
    Then the following three statements are equivalent.
    \begin{enumerate}
        \item [\rm (1)] $\PP(\C_1,\C_2)$ is a symplectic SO $[2n,k_1+k_2,\min\{d_1,d_2\}]_q^{\rm s}$ code.
        \item [\rm (2)] $\PP(\C_1,\C_2)^{\perp_{\rm s}}$ is a symplectic DC 
        $[2n,2n-k_1-k_2,\min\{d_{\rm H}(\C_1^{\perp_{\rm E}}),d_{\rm H}(\C_2^{\perp_{\rm E}})\}]_q^{\rm s}$ code.
        \item [\rm (3)] $\C_1\subseteq \C_2^{\perp_{\rm E}}$.
    \end{enumerate}
\end{theorem}
\begin{proof}
First, we prove that $\PP(\C_1,\C_2)$ and $\PP(\C_1,\C_2)^{\perp_{\rm s}}$ have parameters $[2n,k_1+k_2,\min\{d_1,d_2\}]_q^{\rm s}$ and 
$[2n,2n-k_1-k_2,\min\{d_{\rm H}(\C_1^{\perp_{\rm E}}),d_{\rm H}(\C_2^{\perp_{\rm E}})\}]_q^{\rm s}$, respectively.
Then we reduce the three statements to three simpler statements.
With Definition \ref{def.Plotkin sum}, $\PP(\C_1,\C_2)$ is a $[2n,k_1+k_2]_q$ linear code, 
and then $\PP(\C_1,\C_2)^{\perp_{\rm s}}$ is a $[2n,2n-k_1-k_2]_q$ linear code.
Hence, it suffices to prove that $d_{\rm s}(\PP(\C_1,\C_2))=\min\{d_1,d_2\}$ and 
$d_{\rm s}(\PP(\C_1,\C_2))^{\perp_{\rm s}}=\min\{d_{\rm H}(\C_1^{\perp_{\rm E}}),d_{\rm H}(\C_2^{\perp_{\rm E}})\}$.  \vspace{2mm}

\textbf{Case 1: the proof of $d_{\rm s}(\PP(\C_1,\C_2))=\min\{d_1,d_2\}$.} Let $(\uuu,\uuu+\vvv)\in \PP(\C_1,\C_2)$ 
be any nonzero codeword with $\uuu\in \C_1$ and $\vvv\in \C_2$.
Then $\uuu\neq \mathbf{0}$ or $\vvv\neq \mathbf{0}$. We have two subcases.
    \begin{itemize}
        \item If $\vvv=\mathbf{0},$ since $\uuu$ is a nonzero codeword in $\C_1$,  we have that 
        $$\wt_{\rm s}((\uuu,\uuu+\vvv))=\wt_{\rm s}((\uuu,\uuu))=\wt_{\rm H}(\uuu)\geq d_1.$$

        \item If $\vvv\neq \mathbf{0},$ let $\uuu=(u_1,u_2,\ldots,u_{n})$ and $\vvv=(v_1,v_2,\ldots,v_n)$.
        Note that for each integer $1\leq i\leq n$, if $v_{i}\neq 0$, then $(u_{i},u_{i}+v_{i})\neq (0,0)$ whether $u_{i}=0$ or not;
        if $v_i=0$, then $(u_i,u_i+v_i)=(0,0)$ if and only if $u_i=0$. Hence, we have that 
        $$\wt_{\rm s}((\uuu,\uuu+\vvv))\geq \wt_{\rm H}(\vvv)\geq d_2.$$
    \end{itemize}
    Then it follows that $d_{\rm s}(\PP(\C_1,\C_2))\geq \min\{d_1,d_2\}$.

    Conversely, for $i=1,2$, since $d_{\rm H}(\C_i)=d_i$, there is a codeword $\ccc_i\in \C_i$ such that $\wt_{\rm H}(\ccc_i)=d_i$.
    Note that both $(\ccc_1,\ccc_1)$ and $(\mathbf{0},\ccc_2)$ are codewords in $\PP(\C_1,\C_2)$. Then
    \begin{align*}
        \begin{split}
            d_{\rm s}(\PP(\C_1,\C_2)) & \leq \min\{\wt_{\rm s}((\ccc_1,\ccc_1)),\wt_{\rm s}((\mathbf{0},\ccc_2))\} \\
                     & =\min\{\wt_{\rm H}(\ccc_1),\wt_{\rm H}(\ccc_2)\} \\
                     & = \min\{d_1,d_2\}.
        \end{split}
    \end{align*}
    In summary, $d_{\rm s}(\PP(\C_1,\C_2))=\min\{d_1,d_2\}$. This completes the proof of \textbf{Case 1.} \vspace{2mm}

\textbf{Case 2: the proof of $d_{\rm s}(\PP(\C_1,\C_2))^{\perp_{\rm s}}=\min\{d_{\rm H}(\C_1^{\perp_{\rm E}}),d_{\rm H}(\C_2^{\perp_{\rm E}})\}$.} 
From Definition \ref{def.Plotkin sum}, $\PP(\C_1,\C_2)^{\perp_{\rm E}}$ has a generator matrix $H=\left(\begin{array}{cc}
        H_1 & O \\
        -H_2 & H_2
    \end{array}\right).$ On the other hand, from the definitions of Euclidean dual codes and symplectic dual codes,
    we can deduce that $\PP(\C_1,\C_2)^{\perp_{\rm s}}=\PP(\C_1,\C_2)^{\perp_{\rm E}}\Omega_n=\{\bc \Omega_n \mid \bc\in \PP(\C_1,\C_2)^{\perp_{\rm E}}\}$.
    This turns out that
    $$H\Omega_n=\left(\begin{array}{cc}
        H_1 & O \\
        -H_2 & H_2
    \end{array}\right)\left(\begin{array}{cc}
        O & I_n \\
        -I_n & O
    \end{array}\right)=\left(\begin{array}{cc}
        O & H_1 \\
        -H_2 & -H_2
    \end{array}\right)$$
    is a generator matrix of $\PP(\C_1,\C_2)^{\perp_{\rm s}}$.
    Since $\PP(\C_1,\C_2)^{\perp_{\rm s}}$ is linear and $H_i$ is a generator matrix of $\C_i^{\perp_{\rm E}}$ for $i=1,2$, it is again inferred from
    Definition \ref{def.Plotkin sum} that $\PP(\C_1,\C_2)^{\perp_{\rm s}}$ can be viewed as the Plotkin sum of $\C_2^{\perp_{\rm E}}$
    and $\C_1^{\perp_{\rm E}}$, $i.e.$, $\PP(\C_1,\C_2)^{\perp_{\rm s}}=\PP(\C_2^{\perp_{\rm E}}, \C_1^{\perp_{\rm E}})$.
    Hence, $\PP(\C_1,\C_2)^{\perp_{\rm s}}$ has parameters 
    $[2n,2n-k_1-k_2,\min\{d_{\rm H}(\C_1^{\perp_{\rm E}}),d_{\rm H}(\C_2^{\perp_{\rm E}})\}]_q^{\rm s}$
    according to \textbf{Case 1}. This completes the proof of \textbf{Case 2}.  \vspace{2mm}

By the results above, we only need to prove that the following three simpler statements are equivalent:  
(1') $\PP(\C_1,\C_2)$ is a symplectic SO code; 
(2') $\PP(\C_1,\C_2)^{\perp_{\rm s}}$ is a symplectic DC code; 
and (3') $\C_1\subseteq \C_2^{\perp_{\rm E}}$.   \vspace{2mm}

(1') $\Leftrightarrow$ (2') Since $(\PP(\C_1,\C_2)^{\perp_{\rm s}})^{\perp_{\rm s}}=\PP(\C_1,\C_2)$, we have $\PP(\C_1,\C_2)\subseteq \PP(\C_1,\C_2)^{\perp_{\rm s}}$
if and only if $(\PP(\C_1,\C_2)^{\perp_{\rm s}})^{\perp_{\rm s}}\subseteq \PP(\C_1,\C_2)^{\perp_{\rm s}}$.
Hence $\PP(\C_1,\C_2)$ is symplectic SO if and only if $\PP(\C_1,\C_2)^{\perp_{\rm s}}$ is symplectic DC.

(1') $\Leftrightarrow$ (3') By Definition \ref{def.Plotkin sum}, $\PP(\C_1,\C_2)$ has a generator matrix $G=\left(\begin{array}{cc}
        G_1 & G_1 \\
        O & G_2
\end{array}\right).$ Then we have
\begin{equation}\label{eq.GOmegaGT}
        \begin{split}
        G\Omega_n G^T & = \left( \begin{array}{cc}
        G_1 & G_1 \\
        O & G_2 \\
        \end{array} \right)
        \left( \begin{array}{cc}
            O & I_n \\
            -I_n & O \\
        \end{array} \right)
        \left( \begin{array}{cc}
        G_1^T & O \\
        G_1^T & G_2^T \\
        \end{array} \right)
         = \left( \begin{array}{cc}
        O & G_1G_2^T \\
        -G_2G_1^T & O \\
        \end{array} \right).
        \end{split}
\end{equation}

It follows from Lemma \ref{lem.symplectic SO} and Equation (\ref{eq.GOmegaGT}) that 
$\PP(\C_1,\C_2)$ is symplectic SO if and only if $G\Omega_n G^T=O$, 
if and only if $-G_2G_1^T=-(G_1G_2^T)^T=O$, 
if and only if $\C_1\subseteq \C_2^{\perp_{\rm E}}$. 
Therefore, the statements (1') and (3') are equivalent.

In summary, the statements (1), (2) and (3) are equivalent. This completes the proof.
\end{proof}

We next show that symplectic SO codes obtained by Theorem \ref{th.Plotkin sso} are asymptotically good. 
It is therefore an interesting class of linear codes. 

\begin{theorem}\label{th.Plotkin sum SO good}
    Let $q=2$ or $q$ be an odd prime power. Then $q$-ary symplectic SO codes derived from Theorem \ref{th.Plotkin sso} 
    are asymptotically good with respect to the symplectic distance.  
\end{theorem}
\begin{proof}
    It follows from \cite[Remark 3]{D2009SOgood1} and \cite[Corollary 4.2]{ZC2022SOgood2} that there exist asymptotically good binary 
    and $q$-ary ($q$ is odd) Euclidean SO codes with respect to the Hamming distance.  
    Hence, we assume that $\{\C_i\}_{i=0}^{\infty}$ is an infinite subset of $[n_i,k_i,d_i]_q^{\rm H}$ codes from asymptotically good 
    $q$-ary Euclidean SO codes, where $q=2$ or $q$ is odd.  
    Then it implies that $\lim_{i \rightarrow \infty}n_i=\infty$ with $\liminf_{i\rightarrow \infty}\frac{k_i}{n_i}>0$ 
    and $\liminf_{i\rightarrow \infty}\frac{d_i}{n_i}>0$. 

    Since $\C_i$ is Euclidean SO, $i.e.$, $\C_i\subseteq \C_i^{\perp_{\rm E}}$, it follows form Theorem \ref{th.Plotkin sso} that 
    $\PP(\C_i,\C_i)$ is a $[2n_i,2k_i,d_i]_q^{\rm s}$ symplectic SO code, denoted by $\PP_i$.  
    Note that $\lim_{i \rightarrow \infty}2n_i=\infty$. 
    Also for the infinite sequence of symplectic SO codes $\{\PP_i\}_{i=0}^\infty$,  
    the asymptotic rate of $\{\PP_i\}_{i=0}^\infty$ is 
    \begin{align}
        r=\liminf_{i\rightarrow \infty}\frac{2k_i}{2n_i}=\liminf_{i\rightarrow \infty}\frac{k_i}{n_i}>0
    \end{align}
    and the asymptotic relative symplectic distance of $\{\PP_i\}_{i=0}^\infty$ is 
    \begin{align}
        \delta=\liminf_{i\rightarrow \infty}\frac{d_i}{2n_i}=\frac{1}{2}\liminf_{i\rightarrow \infty}\frac{d_i}{n_i}>0.  
    \end{align}

    This proves the expected result. 
\end{proof}


\subsection{The first construction related to $\ell$-intersection pairs of linear codes}\label{sec.3.2 sso l-intersection}

Let $\C_i$ be an  $[n,k_i]_q$ linear code for $i=1,2$.
Then they are said to be an {\em $\ell$-intersection pair} if $\dim(\C_1\cap \C_2)=\ell$.
Here we recall some important results on $\ell$-intersection pairs of linear codes.

\begin{lemma}[Lemma 2.2 in \cite{GGJT2020l-intersection}]\label{lem.l-intersection111}
    Let $\C_i$ be an  $[n,k_i]_q$ linear code for $i=1,2$. Suppose that $\dim(\C_1 \cap \C_2)=\ell$.
    Then $\max\{k_1+k_2-n, 0\}\leq \ell \leq \min\{k_1, k_2\}.$
\end{lemma}

\begin{lemma}[Theorem 7 in \cite{HFF2022l-intesection222}]\label{lem.l-intersection}
    Let $q\geq 3$ be a prime power. Then there exist two Hamming MDS codes $\C_1$ and $\C_2$ with parameters
    $[n,k_1,n-k_1+1]_q^{\rm H}$ and $[n,k_2,n-k_2+1]_q^{\rm H}$ such that $\dim(\C_1\cap \C_2)=\ell$ if any one of the
    following conditions holds:
    \begin{itemize}
        \item [\rm (1)] $2\leq n\leq q+1$, $1\leq k_1,k_2\leq n-1$, $\max\{k_1+k_2-n,0\}\leq \ell\leq \min\{k_1,k_2\}$
        and $(n,k_1,k_2,\ell)\neq (q+1,2,1,1)$ or $(q+1,1,2,1);$

        \item [\rm (2)] $q=2^m\geq 4$, $0\leq \ell \leq 3$, $n=q+2$ and $(k_1,k_2)=(3,q-1)$, $(q-1,3)$ or $(3,3);$

        \item [\rm (3)] $q=2^m\geq 4$, $q-4\leq \ell \leq q-1$, $n=q+2$ and $(k_1,k_2)=(q-1,q-1)$.
    \end{itemize}
\end{lemma}

Next, we apply the Plotkin sum construction to $\ell$-intersection pairs of linear codes. 
By Theorem \ref{th.Plotkin sso}, we can transform $\ell$-intersection pairs of linear codes 
to symplectic SO codes, which is not effectively applicable to the Euclidean and Hermitian inner 
products under the Hamming distance. 
In addition, these constructions have explicit parameters.

\begin{theorem}\label{th.symplectic SO.l-intersection}
    Let $q\geq 3$ be a prime power. Let $n$, $k_1$ and $k_2$ be any three positive integers satisfying
    $2\leq n\leq q+1$, $1\leq k_1\leq k_2\leq n-1$ and $k_1+k_2\geq n$.
    Then the following statements hold.
    \begin{enumerate}
        \item [\rm (1)]  There exists a symplectic SO $[2n,n+k_1-k_2,n-k_1+1]_q^{\rm s}$ code.

        \item [\rm (2)] There exists a symplectic DC $[2n,n+k_2-k_1,n-k_2+1]_q^{\rm s}$ code.
    \end{enumerate}
\end{theorem}
\begin{proof}
    For each integer $2\leq n\leq q+1$ and $1\leq k_1\leq k_2\leq n-1$ satisfying $(n,k_1,k_2)\neq (q+1,2,1)$ or $(q+1,1,2)$,
    it follows from Lemma \ref{lem.l-intersection} (1) that there exist two Hamming MDS codes,
    denoted by $\C_1$ and $\C_2$, with respective parameters $[n,k_1,n-k_1+1]_q^{\rm H}$ and $[n,k_2,n-k_2+1]_q^{\rm H}$
    such that $\dim(\C_1\cap \C_2)=k_1$.
    Then $\C_2^{\perp_{\rm E}}$ is a Hamming MDS $[n,n-k_2,k_2+1]_q^{\rm H}$ code and $\C_1\subseteq \C_2=(\C_2^{\perp_{\rm E}})^{\perp_{\rm E}}$.

    Hence, $\PP(\C_1,\C_2^{\perp_{\rm E}})$  is a symplectic SO $[2n,n+k_1-k_2,\min\{n-k_1+1,k_2+1\}]_q^{\rm s}$ code and $\PP(\C_1,\C_2^{\perp_{\rm E}})^{\perp_{\rm s}}$ is a symplectic DC $[2n,n+k_2-k_1,\min\{k_1+1,n-k_2+1\}]_q^{\rm s}$ code from Theorem \ref{th.Plotkin sso}.
    Since $d_{\rm s}(\PP(\C_1,\C_2^{\perp_{\rm E}}))=\min\{n-k_1+1,k_2+1\}$ and
    $d_{\rm s}(\PP(\C_1,\C_2^{\perp_{\rm E}})^{\perp_{\rm s}})=\min\{k_1+1,n-k_2+1\}$,
    we have the following two cases.
    \begin{itemize}
        \item \textbf{Case 1:} If $k_1+k_2\geq n$, we have $d_{\rm s}(\PP(\C_1,\C_2^{\perp_{\rm E}}))=n-k_1+1$ and
        $d_{\rm s}(\PP(\C_1,\C_2^{\perp_{\rm E}})^{\perp_{\rm s}})=n-k_2+1$, $i.e.$,
        a symplectic SO $[2n,n+k_1-k_2,n-k_1+1]_q^{\rm s}$ code and
        a symplectic DC $[2n,n+k_2-k_1,n-k_2+1]_q^{\rm s}$ code exist.

        \item \textbf{Case 2:} If $k_1+k_2\leq n$, we have $d_{\rm s}(\PP(\C_1,\C_2^{\perp_{\rm E}}))=k_2+1$ and
        $d_{\rm s}(\PP(\C_1,\C_2^{\perp_{\rm E}})^{\perp_{\rm s}})=k_1+1$, $i.e.$,
        a symplectic SO $[2n,n+k_1-k_2,k_2+1]_q^{\rm s}$ code and
        a symplectic DC $[2n,n+k_2-k_1,k_1+1]_q^{\rm s}$ code exist.
    \end{itemize}
Due to the flexibility of the ranges of $k_1$ and $k_2$, it is not difficult to check that \textbf{Case 1} and \textbf{Case 2} 
will yield the same families of symplectic SO (resp. DC) codes. Hence, we only consider $\textbf{Case 1}$ for our construction. 
Moreover, since $1\leq k_1\leq k_2\leq n-1$ and $k_1+k_2\geq n$ are limited in $\textbf{Case 1}$, 
then $(n,k_1,k_2)$ are always not equal to $(q+1,2,1)$ or $(q+1,1,2)$. 
Therefore, the desired results (1) and (2) follow.
\end{proof}

\begin{theorem}\label{th.symplectic SO and DC MDS codes}
    Let $q\geq 3$ be a prime power. Let $2\leq n\leq q+1$ be a positive integer.
    Then the following statements hold.
    \begin{enumerate}
        \item [\rm (1)] There exists a symplectic MDS symplectic SO $[2n,2k,n-k+1]_q^{\rm s}$ code
        and a symplectic MDS symplectic DC $[2n,2n-2k,k+1]_q^{\rm s}$ code for each integer $1\leq k\leq \left \lfloor \frac{n}{2} \right \rfloor$.

        \item [\rm (2)] There exists a symplectic MDS symplectic SO $[2n,2k+1,n-k]_q^{\rm s}$ code
        and a symplectic MDS symplectic DC $[2n,2n-2k-1,k+1]_q^{\rm s}$ code for each integer $0\leq k\leq \left \lfloor \frac{n-1}{2} \right \rfloor$.

        \item [\rm (3)] There exists a symplectic MDS symplectic SO $[2^{m+1}+4,6,2^m]_{2^m}^{\rm s}$ code and a symplectic MDS symplectic DC $[2^{m+1}+4,2^{m+1}-2,4]_{2^m}^{\rm s}$ code for each $m\geq 2$.
    \end{enumerate}
\end{theorem}
\begin{proof}
    (1) Let $k_1+k_2=n$ and $1\leq k_1\leq \left \lfloor \frac{n}{2} \right \rfloor$ in Theorem \ref{th.symplectic SO.l-intersection}.
    Then we can obtain a symplectic SO $[2n,2k_1,n-k_1+1]_q^{\rm s}$ code and a symplectic DC $[2n,2n-2k_1,k_1+1]_q^{\rm s}$ code.
    It follows from the facts $n-k_1+1=\left \lfloor \frac{2n-2k_1+2}{2} \right \rfloor$
    and $k_1+1=\left \lfloor \frac{2n-(2n-2k_1)+2}{2} \right \rfloor$ that these two codes are symplectic MDS.
    This completes the proof of (1) by taking $k=k_1$.

    (2) Let $k_1+k_2=n+1>n$ and $1\leq k_1\leq \left \lfloor \frac{n+1}{2} \right \rfloor$ in Theorem \ref{th.symplectic SO.l-intersection}.
    Then we can obtain a symplectic SO $[2n,2k_1-1,n-k_1+1]_q^{\rm s}$ code and a symplectic DC $[2n,2n-2k_1+1,k_1]_q^{\rm s}$ code.
    It follows from the facts $n-k_1+1=\left \lfloor \frac{2n-(2k_1-1)+2}{2} \right \rfloor$
    and $k_1=\left \lfloor \frac{2n-(2n-2k_1+1)+2}{2} \right \rfloor$
    that these two codes are symplectic MDS. This completes the proof of (2) by taking $k=k_1-1$.

    (3) Similar to the proof of Theorem \ref{th.symplectic SO.l-intersection}, for each $m\geq 2$,
    it follows from taking $(k_1,k_2,\ell)=(3,2^m-1,3)$ in Lemma \ref{lem.l-intersection} (2) and
    Theorem \ref{th.Plotkin sso} that a symplectic SO $[2^{m+1}+4,6,2^m]_{2^m}$ code and
    a symplectic DC $[2^{m+1}+4,2^{m+1}-2,4]_{2^m}$ code exist.
    Similar to (1) above, it is easy to check that these two codes are symplectic MDS.
    This completes the desired result (3).
\end{proof}

\begin{theorem}\label{th.symplectic self-dual}
    There exists a symplectic self-dual $[2^{m+1}+4,2^m+2,4]_{2^m}^{\rm s}$ code for each integer $m\geq 2$.
\end{theorem}
\begin{proof}
Similar to the proof of Theorem \ref{th.symplectic SO.l-intersection} again, the desired result follows by taking $(k_1,k_2,\ell)=(3,3,3)$ in Lemma \ref{lem.l-intersection} (2) or by taking $(k_1,k_2,\ell)=(2^m-1,2^m-1,2^m-1)$ in Lemma \ref{lem.l-intersection} (3).
\end{proof}

In particular, symplectic MDS symplectic self-dual codes can be further deduced from
Theorems \ref{th.symplectic SO and DC MDS codes} and \ref{th.symplectic self-dual} as follows.

\begin{corollary}\label{coro.symplectic self-dual codes}
    Let $q\geq 3$ be a prime power. 
    Then the following statements hold.
    \begin{enumerate}
        \item [\rm (1)] There exists a symplectic MDS symplectic self-dual $[2n,n,\frac{n}{2}+1]_q^{\rm s}$ code for each even $2\leq n\leq q+1$.

        \item [\rm (2)] There exists a symplectic MDS symplectic self-dual $[2n,n,\frac{n+1}{2}]_q^{\rm s}$ code for each odd $2\leq n\leq q+1$.

        \item [\rm (3)] There exists a symplectic MDS symplectic self-dual $[12,6,4]_4^{\rm s}$ code.
    \end{enumerate}
\end{corollary}
\begin{proof}
    Taking $k=\frac{n}{2}$ with even $n$ and $\frac{n-1}{2}$ with odd $n$ in Theorems \ref{th.symplectic SO and DC MDS codes} (1) and (2), respectively,
    the desired results (1) and (2) follow immediately. Also the desired result (3) holds by taking $m=2$ in Theorem \ref{th.symplectic SO and DC MDS codes} (3)
    or Theorem \ref{th.symplectic self-dual}.
\end{proof}

\begin{remark}
As mentioned earlier, this construction is not effectively applicable to the Euclidean and Hermitian inner products 
under the Hamming distance based on $\ell$-intersection pairs of linear codes. 
Currently, almost all known Hamming MDS Euclidean and Hermitian SO codes are 
constructed from (extended) generalized Reed-Solomon codes. 
These known Hamming MDS SO codes can yield additive (in fact, linear) SO codes  
of generalized Reed-Solomon type by \cite[Lemma 18]{quantum-codes-IT-2}.
It is not difficult to check that symplectic SO codes $\C$ we obtain from $\ell$-intersection pairs of linear codes 
are not equivalent to known Hamming MDS Hermitian SO codes, 
even if they have the same parameters because our method is more universal and $\phi(\C)$ is nonlinear additive. 
In fact, Adriaensen and Ball \cite{Add-MDS} also motivated us to construct additive MDS codes with length less than or equal to $q+1$.
\end{remark}

\begin{example}
    Let $q=8$. By Part (1) of Theorem \ref{th.symplectic SO and DC MDS codes}, we can obtain symplectic MDS symplectic SO
    $[6,2,3]_8^{\rm s}$, $[8,2,4]_8^{\rm s}$, $[10,2,5]_8^{\rm s}$, $[10,4,4]_8^{\rm s}$, $[12,2,6]_8^{\rm s}$, $[12,4,5]_8^{\rm s}$,
    $[14,2,7]_8^{\rm s}$, $[14,4,6]_8^{\rm s}$, $[14,6,5]_8^{\rm s}$, $[16,2,8]_8^{\rm s}$, $[16,4,7]_8^{\rm s}$, $[16,6,6]_8^{\rm s}$,
    $[18,2,9]_8^{\rm s}$, $[18,4,8]_8^{\rm s}$, $[18,6,7]_8^{\rm s}$, $[18,8,6]_8^{\rm s}$ and $[20,6,8]_8^{\rm s}$ codes.
    By Part (2) of Theorem \ref{th.symplectic SO and DC MDS codes}, we can obtain symplectic MDS symplectic SO $[8,3,3]_8^{\rm s}$, $[10,3,4]_8^{\rm s}$, $[12,3,5]_8^{\rm s}$, $[12,5,4]_8^{\rm s}$, $[14,3,6]_8^{\rm s}$, $[14,5,5]_8^{\rm s}$,
    $[16,3,7]_8^{\rm s}$, $[16,5,6]_8^{\rm s}$, $[16,7,5]_8^{\rm s}$, $[18,3,8]_8^{\rm s}$, $[18,5,7]_8^{\rm s}$ and $[18,7,6]_8^{\rm s}$ codes.
    By Part (1) of Corollary \ref{coro.symplectic self-dual codes}, we can obtain
    symplectic MDS symplectic self-dual $[4,2,2]_8^{\rm s}$, $[8,4,3]_8^{\rm s}$, $[12,6,4]_8^{\rm s}$ and $[16,8,5]_8^{\rm s}$ codes. 
    By Part (2) of Corollary \ref{coro.symplectic self-dual codes}, we can obtain
    symplectic MDS symplectic self-dual $[6,3,2]_8^{\rm s}$, $[10,5,3]_8^{\rm s}$, $[14,7,4]_8^{\rm s}$ and $[18,9,5]_8^{\rm s}$ codes.
    In addition, Theorem \ref{th.symplectic self-dual} also yields a symplectic self-dual $[20,10,4]_8^{\rm s}$ code.
\end{example}

\subsection{The second construction related to generalized Reed-Muller codes}\label{sec.3.3 sso GRM}

We recall some notions on {\em generalized Reed-Muller} (GRM) codes and refer to \cite{Huffman} for more details. Let $R=\F_q[x_1,x_2,\ldots,x_m]$ be the ring of polynomials over $\F_q$
and $I=\langle  x_1^q-x_1,x_2^q-x_2,\ldots,x_m^q-x_m \rangle$ be the ideal of $R$.
Denote the corresponding $\F_q$-algebra by $AL=\F_q[x_1,x_2,\ldots,x_m]/I$ and the set of
zeros in $\F_q$ of $I$ by $O(I)=\F_q^m=\{P_1,P_2,\ldots,P_n\}$.
Then for any integer $r\geq 0$, the $r$-th order GRM code $\GRM(r,m)$ is defined as
\begin{align}\label{eq.GRMdef}
    \GRM(r,m)=\{(f(P_1),f(P_2),\ldots,f(P_n))\mid f\in AL, \deg(f)\leq r\}
\end{align}
and we always select a canonical representative of $f$ excluding power $x^j_i$ for $j\geq q$.
Then some known results of GRM codes from \cite{Huffman} are shown as follows.

\begin{lemma}[\cite{Huffman}]\label{lem.GRM properties}
    Let notations be the same as above. Suppose $0\leq r< m(q-1)$ and write
    $m(q-1)-r=a(q-1)+b$ with integers $a,b\geq 0$ and $b<q-1$. Then the following statements hold.
    \begin{enumerate}
        \item [\rm (1)] $\GRM(r,m)$ has parameters
        \begin{align}
            \left[q^m, \sum_{j=0}^{m}(-1)^j\binom{m}{j}\binom{m+r-jq}{r-jq}, (b+1)q^a\right]_q.
        \end{align}

        \item [\rm (2)] The Euclidean dual code of a GRM code is still a GRM code.
        Specifically, we have
        \begin{align}
            \GRM(r,m)^{\perp_{\rm E}}=\GRM(m(q-1)-r-1,m).
        \end{align}

        \item [\rm (3)] For any integers $0\leq i\leq r< m(q-1)$, we have
        \begin{align}
            \GRM(i,m)\subseteq \GRM(r,m).
        \end{align}
    \end{enumerate}
\end{lemma}

Note that Theorem \ref{th.Plotkin sso} generally gives $q$-ary symplectic SO (resp. DC) codes of length up to $2q+2$
(some can take up to $2q+4$ when $q=2^m\geq 4$).
Next, we apply the Plotkin sum construction to GRM codes. This permits us to construct $q$-ary symplectic SO (resp. DC)
codes of length exceeding $2q+2$.

\begin{theorem}\label{th.SSO from GRM}
    Let $r,i$ and $m$ be three positive integers satisfying $0\leq r<m(q-1)$ and $0\leq i\leq m(q-1)-r-1$.
    Write $m(q-1)-r=a_r(q-1)+b_r$, $m(q-1)-i=a_i(q-1)+b_i$, $r+1=a_r'(q-1)+b_r'$ and $i+1=a_i'(q-1)+b_i'$,
    where $a_r, a_i, a_r', a_i'\geq 0$ and $0\leq b_r,b_i,b_r',b_i'< q-1$.
    Then the following statements hold.
    \begin{enumerate}
        \item [\rm (1)] There exists a symplectic SO code with parameters
        \begin{small}
            $$\left[ 2q^m, \sum_{j=0}^m(-1)^j\binom{m}{j}\left[\binom{m+i-jq}{i-jq}+\binom{m+r-jq}{r-jq}\right], \min\left\{(b_i+1)q^{a_i}, (b_r+1)q^{a_r}\right\} \right]_q^{\rm s}.$$
        \end{small}

        \item [\rm (2)] There exists a symplectic DC code with parameters
        \begin{footnotesize}
            $$\left[ 2q^m, 2q^m-\sum_{j=0}^m(-1)^j\binom{m}{j}\left[\binom{m+i-jq}{i-jq}+\binom{m+r-jq}{r-jq}\right], \min\left\{(b_i'+1)q^{a_i'}, (b_r'+1)q^{a_r'}\right\} \right]_q^{\rm s}.$$
        \end{footnotesize}
    \end{enumerate}
\end{theorem}
\begin{proof}
    (1) Let notations be the same as before. Let $\C_2=\GRM(r,m)$.
    Since $r$ is a positive integer satisfying $0\leq r<m(q-1)$,
    then it follows from Lemma \ref{lem.GRM properties} (1) and the representation of $r$ that $\C_2$ has parameters
    $[q^m, \sum_{j=0}^m(-1)^j\binom{m}{j}\binom{m+r-jq}{r-jq},(b_r+1)q^{a_r}]_q^{\rm H}$.
    From Lemma \ref{lem.GRM properties} (2), we have $\C_2^{\perp_{\rm E}}=\GRM(m(q-1)-r-1,m)$.

    Note that $i$ is a positive integer satisfying $0\leq i\leq m(q-1)-r-1$, then $0\leq i< m(q-1)$.
    On one hand, from Lemma \ref{lem.GRM properties} (1) and the representation of $i$, $\GRM(i,m)$ has parameters
    $$\left[q^m, \sum_{j=0}^m(-1)^j\binom{m}{j}\binom{m+i-jq}{i-jq},(b_i+1)q^{a_i}\right]_q^{\rm H}.$$
    On the other hand, from Lemma \ref{lem.GRM properties} (3), we have $\GRM(i,m)\subseteq \GRM(m(q-1)-r-1,m)=\C_2^{\perp_{\rm E}}$.
    Let $\C_1=\GRM(i,m)$. Then $\C_1\subseteq \C_2^{\perp_{\rm E}}$. Therefore, from Theorem \ref{th.Plotkin sso} (1), $\PP(\C_1,\C_2)$ is a symplectic SO code with parameters
    \begin{small}
        $$\left[ 2q^m, \sum_{j=0}^m(-1)^j\binom{m}{j}\left[\binom{m+i-jq}{i-jq}+\binom{m+r-jq}{r-jq}\right], \min\left\{(b_i+1)q^{a_i}, (b_r+1)q^{a_r}\right\} \right]_q^{\rm s}$$
    \end{small}
This completes the proof of (1).

    (2) Let notations be the same as (1) above.
    From Lemma \ref{lem.GRM properties} (2), we have $\C_1^{\perp_{\rm E}}=\GRM(m(q-1)-i-1,m)$.
    Since $0\leq r<m(q-1)$ and $0\leq i\leq m(q-1)-r-1$, we have $0\leq m(q-1)-r-1<m(q-1)$ and $0\leq r\leq  m(q-1)-i-1<m(q-1)$.
    It follows from Lemma \ref{lem.GRM properties} (1) and the representations of $r+1$ and $i+1$
    that the respective minimum Hamming distances of $\C_1^{\perp_{\rm E}}$ and $\C_2^{\perp_{\rm E}}$ are $(b_i'+1)q^{a_i'}$ and $(b_r'+1)q^{a_r'}$.
    Therefore, from Theorem \ref{th.Plotkin sso} (2), $\PP(\C_1,\C_2)^{\perp_{\rm s}}$ is a symplectic DC code with parameters
    \begin{small}
        $$\left[ 2q^m, 2q^m-\sum_{j=0}^m(-1)^j\binom{m}{j}\left[\binom{m+i-jq}{i-jq}+\binom{m+r-jq}{r-jq}\right], \min\left\{(b_i'+1)q^{a_i'}, (b_r'+1)q^{a_r'}\right\} \right]_q^{\rm s}$$
    \end{small}
This completes the proof of (2).
\end{proof}

\begin{example}
 We list some symplectic SO and DC codes in Table 1 obtained by Theorem \ref{th.SSO from GRM}.
 Note that the differences between the symplectic Singleton bound and their minimum symplectic distances are at most $2$ and some of them are $0$. Furthermore, it is well-known that the Singleton bound is a rough bound when the length is relatively large compared to $q$. Hence, these codes listed indeed have good parameters.
\end{example}

  \begin{center}
  \begin{threeparttable}
    \begin{tabular}{c|c|c|c|c|c||c|c|c|c|c|c}
    \multicolumn{12}{c}{Table 1: Some good symplectic SO and DC codes from Theorem \ref{th.SSO from GRM} }\\
     \hline
      $q$ & $m$ & $r$ & $i$ & Symplectic SO & $\text {Bound}$  & $q$ & $m$ & $r$ & $i$ & Symplectic DC & Bound  \\ \hline\hline
         2 & 2 & 0 & 1 & $[8,4,2]_2^{\rm s}$ & 3 &3 & 2 & 1 & 1 & $[18,12,3]_3^{\rm s}$ & 4  \\
         2 & 3 & 1 & 1 & $[16,8,4]_2^{\rm s}$ & 5& 3 & 3 & 1 & 1 & $[54,46,3]_3^{\rm s}$ & 5 \\

         3 & 1 & 1 & 0 & $[6,3,2]_3^{\rm s}$ & 2 &4 & 2 & 1 & 1 & $[32,26,3]_4^{\rm s}$ & 4\\

         3 & 2 & 1 & 1 &  $[18,6,6]_3^{\rm s}$ & 7 &4 & 3 & 1 & 1 & $[128,120,3]_4^{\rm s}$ & 5\\
         3 & 2 & 2 & 1 &  $[18,9,3]_3^{\rm s}$ & 5 &5 & 1 & 1 & 1 & $[10,6,3]_5^{\rm s}$ & 3\\

         4 & 1 & 1 & 1 & $[8,4,3]_4^{\rm s}$ &  3 &5 & 2 & 1 & 1 & $[50,44,3]_5^{\rm s}$ & 4\\
         4 & 2 & 1 & 1 & $[32,6,12]_4^{\rm s}$  & 14 &5 & 3 & 1 & 1 & $[250,242,3]_5^{\rm s}$ & 5\\

         5 & 1 & 1 & 1 & $[10,4,4]_5^{\rm s}$ & 4 &7 & 1 & 1 & 1 & $[14,10,3]_7^{\rm s}$ & 3\\
         5 & 1 & 2 & 1 & $[10,5,3]_5^{\rm s}$ & 3 &7 & 1 & 2 & 2 & $[14,8,4]_7^{\rm s}$  & 4 \\

         7 & 1 & 1 & 1 & $[14,4,6]_7^{\rm s}$ & 6 &7 & 2 & 1 & 1 & $[98,92,3]_7^{\rm s}$ & 4\\

         7 & 1 & 2 & 2 & $[14,6,5]_7^{\rm s}$ & 5 &7 & 3 & 1 & 1 & $[686,678,3]_7^{\rm s}$ & 5\\
         \hline
   \end{tabular}
   \begin{tablenotes}
    \footnotesize
    \item The ``Bound'' in the sixth and twelfth columns denote the symplectic Singleton bound.
    \end{tablenotes}
    \end{threeparttable}
\end{center}

\section{Symplectic LCD codes}\label{sec4}

In this section, we construct symplectic LCD codes from the Plotkin sum construction. 
The following criterion is important for us. 

\begin{theorem}\label{th.Plotkin SLCD}
	Let $\C_i$ be an $[n,k_i,d_{i}]_q^{\rm H}$ linear code for $i=1, 2$.
    Then the following three statements are equivalent.
    \begin{enumerate}
        \item [\rm (1)] $\PP(\C_1,\C_2)$ is a symplectic LCD $[2n,k_1+k_2,\min\{d_1,d_2\}]_q^{\rm s}$ code.
        \item [\rm (2)] $\PP(\C_1,\C_2)^{\perp_{\rm s}}$ is a symplectic LCD 
        $[2n,2n-k_1-k_2,\min\{d_{\rm H}(\C_1^{\perp_{\rm E}}),d_{\rm H}(\C_2^{\perp_{\rm E}})\}]_q^{\rm s}$ code.
        \item [\rm (3)] $\C_1\cap \C_2^{\perp_{\rm E}}=\{\mathbf{0}\}$ and $k_1=k_2$.
    \end{enumerate}
\end{theorem}

\begin{proof}
The parameters of $\PP(\C_1,\C_2)$ and $\PP(\C_1,\C_2)^{\perp_{\rm s}}$ are straightforward from  the proof of Theorem \ref{th.Plotkin sso}.

(1) $\Leftrightarrow$ (2) The result is obvious since $(\PP(\C_1,\C_2)^{\perp_{\rm s}})^{\perp_{\rm s}}=\PP(\C_1,\C_2)$.

    (1) $\Leftrightarrow$ (3) By Equation (\ref{eq.GOmegaGT}), $\PP(\C_1,\C_2)$ is symplectic LCD if and only if ${\rm rank}(G\Omega_n G^T)=2{\rm rank}(G_1G_2^T)=k_1+k_2$.
    Since ${\rm rank}(G_1G_2^T)=k_1-\dim(\C_1\cap \C_2^{\perp_{\rm E}})$, then $\PP(\C_1,\C_2)$ is symplectic LCD 
    is further equivalent to $2\dim(\C_1\cap \C_2^{\perp_{\rm E}})=k_1-k_2$.
    Note that $\C_1$ and $\C_2^{\perp_{\rm E}}$ have respective parameters $[n,k_1]_q$ and $[n,n-k_2]_q$. 
    Then according to Lemma \ref{lem.l-intersection111},
    we have $2\dim(\C_1\cap \C_2^{\perp_{\rm E}})\geq \max\{2k_1-2k_2, 0\}$.
    Hence, $2\dim(\C_1\cap \C_2^{\perp_{\rm E}})=k_1-k_2$ holds if and only if
    $k_1-k_2\geq 0$ and $k_1-k_2\geq 2k_1-2k_2$, $i.e.$, $\dim(\C_1\cap \C_2^{\perp_{\rm E}})=0$ ($i.e.$, $\C_1\cap \C_2^{\perp_{\rm E}}=\{\mathbf{0}\}$) and $k_1=k_2$,
    which implies that the statements (1) and (3) are equivalent.

    In summary, the statements (1), (2) and (3) are equivalent. This completes the proof.
\end{proof}

Similar to Theorem \ref{th.Plotkin sum SO good}, we demonstrate that symplectic LCD codes derived from Theorem \ref{th.Plotkin SLCD} 
also exhibit asymptotic goodness. Therefore, these symplectic LCD codes are also meaningful. 

\begin{theorem}\label{th.Plotkin sum LCD good}
    Let $q$ be a prime power. Then $q$-ary symplectic LCD codes derived from 
    Theorem \ref{th.Plotkin SLCD}  are asymptotically good with respect to the symplectic distance. 
\end{theorem}
\begin{proof}
    By \cite{M1992}, $q$-ary Euclidean LCD codes are asymptotically good with respect to the Hamming distance.  
    Let $\{\C_i\}_{i=0}^{\infty}$ be an infinite subset of $[n_i,k_i,d_i]_q^{\rm H}$ codes from asymptotically good 
    $q$-ary Euclidean LCD codes. 
    Since $\dim(\C_i\cap \C_i^{\perp_{\rm E}})=0$, it follows from Theorem \ref{th.Plotkin SLCD} that 
    $\PP(\C_i,\C_i)$ gives a $[2n_i,2k_i,d_i]_q^{\rm s}$ symplectic LCD code, denoted by $\PP_i$.   
    Then the rest of the proof is similar to that of Theorem \ref{th.Plotkin sum SO good}. 
    Therefore, the excepted result follows. 
\end{proof}



Furthermore, we have the following result. 

\begin{theorem}\label{th.2-ary symplectic LCD codes}
    Let $\C_1$ be an $[n,k,d]_q^{\rm H}$ linear code. If there is a permutation matrix $P$ such that $\C_2=\C_1P$ and $\C_1\cap \C_2^{\perp_{\rm E}}=\{\mathbf{0}\}$,
    then there is a symplectic LCD $[2n,2k,d]_q^{\rm s}$, $i.e.$, an ACD $(n,q^{2k},d)_{q^2}^{\rm H}$ code
    and a symplectic LCD $[2n,2n-2k,d_{\rm H}(\C^{\perp_{\rm E}})]_q^{\rm s}$, 
    $i.e.$, an ACD $(n,q^{2n-2k},d_{\rm H}(\C^{\perp_{\rm E}}))_{q^2}^{\rm H}$ code.
\end{theorem}
\begin{proof}
    Since $\C_2=\C_1P$ and $P$ is a permutation matrix, then $\C_2$ is also an $[n,k,d]_q^{\rm H}$ linear code.
    Then it follows from Theorem \ref{th.Plotkin SLCD} that there exists a symplectic LCD $[2n,2k,d]_q^{\rm s}$ code $\C'$ 
    and a symplectic LCD $[2n,2n-2k,d_{\rm H}(\C^{\perp_{\rm E}})]_q^{\rm s}$ code $\C''$. 
    Consider $\phi(\C')$ and $\phi(\C'')$. 
    The desired result follows immediately from \cite[Lemma 14]{quantum-codes-IT-2}. 
\end{proof}

For convenience, we use an array to represent a permutation matrix. 
For example, $P=(3\ 5\ 2\ 1\ 6\ 4)$ denotes a permutation matrix $P$ 
whose $(1,3), (2,5), (3,2), (4,1), (5,6), (6,4)$-entries are $1$ and others are $0$. 
Note also that a $[2n,2k,d]_2^{\rm s}$ binary symplectic LCD code is further equivalent to 
a so-called $(n,2^{2k},d)_4^{\rm H}$ quaternary trace Hermitian (TrH) ACD code \cite{CRSS1998}. 
For more details on TrH ACD codes, one can refer to recent papers \cite{GLLM2023,SLKS2023,SLOS2022}. 
In the following, we give some specific constructions of binary symplectic LCD codes, $i.e.,$ quaternary TrH ACD codes 
that outperform best-known quaternary Hermitian LCD codes reported in the literature.

\begin{example}\label{exam.SLCD2}
  In \cite{L2018}, Li constructed a best-known $[63,56,3]_4^{\rm H}$ quaternary Hermitian LCD code. 
  Take $\C_1$ as the best-known $[63,56,4]_2^{\rm H}$ binary linear code in the current $\texttt{MAGMA BKLC}$ database \cite{Magma,codetable}. 
  Let 
  \begin{align*}
      \begin{split}
       P_{65}=(& 53\ 29\ 63\ 14\ 44\ 47\ 46\ 4\ 51\ 59\ 11\ 20\ 10\ 23\ 13\ 37\ 42\ 9\ 26\ 34\ 12 \\ 
          & 49\ 38\ 30\ 62\ 56\ 16\ 55\ 28\ 33\ 3\ 61\ 40\ 6\ 5\ 35\ 22\ 24\ 52\ 50\ 25\ 7\ \\ 
          & 18\ 39\ 36\ 31\ 8\ 21\ 27\ 57\ 17\ 60\ 41\ 58\ 19\ 43\ 54\ 48\ 1\ 32\ 15\ 45\ 2) 
      \end{split}
  \end{align*}
  and $\C_2=\C_1P_{65}$. 
  Verified by the Magma software package \cite{Magma}, 
  $\C_1\cap \C_2^{\perp_{\rm E}}=\{\mathbf{0}\}.$ 
  Hence, Theorem \ref{th.2-ary symplectic LCD codes} yields a $[126,112,4]_2^{\rm s}$ binary symplectic LCD code,  
  which is equivalent to  a $(63,2^{56},4)_4^{\rm H}$ TrH ACD code. 
  Note that the $(63,2^{56},4)_4^{\rm H}$ TrH ACD code outperforms the best-known $[63,56,3]_4^{\rm H}$ Hermitian LCD code.  
  In Table 2, we list more quaternary TrH ACD codes derived from Theorem \ref{th.2-ary symplectic LCD codes} outperform linear counterparts. 
\end{example}

\begin{center}
  \begin{threeparttable}
    \begin{tabular}{c|c|c|c|c}
    \multicolumn{5}{c}{Table 2: Some quaternary TrH ACD codes outperform linear counterparts}\\
     \hline
     $\C_1$ & $P$ & $\begin{array}{c}\text{Binary symplectic} \\ \text{LCD codes} \end{array}$ & $\begin{array}{c}\text{Quaternary TrH}\\ \text{ACD codes}\end{array}$ 
     & $\begin{array}{c}\text{Known quaternary}\\ \text{Hermitian LCD codes} \end{array}$    \\ \hline
     
     $[46,23,11]_2^{\rm H}$ & $P_{46}$ & $[92,46,11]_2^{\rm s}$ & $(46,2^{23},11)_4^{\rm H}$ & $[46,23,7]_4^{\rm H}$   \cite{G2020} \\ 
    
     $[52,26,10]_2^{\rm H}$ & $P_{52}$ & $[104,52,10]_2^{\rm s}$ & $(52,2^{26},10)_4^{\rm H}$ & $[52,26,2]_4^{\rm H}$  \cite{G2020} \\ 

     $[56,28,12]_2^{\rm H}$ & $P_{56}$ & $[112,56,12]_2^{\rm s}$ & $(56,2^{28},12)_4^{\rm H}$ & $[56,28,4]_4^{\rm H}$  \cite{G2020} \\  

     $[58,29,12]_2^{\rm H}$ & $P_{58}$ & $[116,58,12]_2^{\rm s}$ & $(58,2^{29},12)_4^{\rm H}$ & $[58,29,4]_4^{\rm H}$  \cite{G2020} \\ 

     $[62,31,12]_2^{\rm H}$ & $P_{62}$ & $[114,62,12]_2^{\rm s}$ & $(62,2^{31},12)_4^{\rm H}$ & $[62,31,7]_4^{\rm H}$  \cite{CGS2023} \\ 

     $[63,56,4]_2^{\rm H}$ & $P_{63}$ & $[126,112,4]_2^{\rm s}$ & $(63,2^{56},4)_4^{\rm H}$ & $[63,56,3]_4^{\rm H}$  \cite{L2018} \\ 

     $[64,32,12]_2^{\rm H}$ & $P_{64}$ & $[128,64,12]_2^{\rm s}$ & $(64,2^{32},12)_4^{\rm H}$ & $[64,32,8]_4^{\rm H}$  \cite{CGS2023} \\ 

     $[70,35,14]_2^{\rm H}$ & $P_{70}$ & $[140,70,14]_2^{\rm s}$ & $(70,2^{35},14)_4^{\rm H}$ & $[70,35,5]_4^{\rm H}$  \cite{G2020} \\ 

     $[72,36,15]_2^{\rm H}$ & $P_{72}$ & $[144,72,15]_2^{\rm s}$ & $(72,2^{36},15)_4^{\rm H}$ & $[72,36,6]_4^{\rm H}$  \cite{G2020} \\ 
     
     $[74,37,14]_2^{\rm H}$ & $P_{74}$ & $[148,74,14]_2^{\rm s}$ & $(74,2^{37},14)_4^{\rm H}$ & $[74,37,4]_4^{\rm H}$  \cite{G2020} \\ 
     \hline 
\end{tabular}
\begin{tablenotes}
\footnotesize
\item $\cdot$ $\C_1$ is the best-known binary linear code in the current $\texttt{MAGMA BKLC}$ database \cite{Magma,codetable}. 
\item $\cdot$ $P_{46}, P_{52}, P_{56}, P_{58}, P_{62}, P_{63}, P_{64}, P_{70}, P_{72}, P_{74}$ are permutation matrices listed in Appendix. 
\end{tablenotes}
\end{threeparttable}        
\end{center}

\section{Conclusions}\label{sec5}

In this paper, we study the Plotkin sum construction with respect to the symplectic inner product and symplectic distance. 
Two criteria for linear codes derived from the Plotkin sum construction being symplectic SO and LCD codes are proposed.  
Symplectic SO and LCD codes constructed by these two ways are also proved to be asymptotically good. 
Based on these criteria, we further present several explicit constructions of symplectic SO and LCD codes via 
$\ell$-intersection pairs of linear codes, GRM codes and general linear codes. 
As a result, many symplectic SO (resp. DC) and LCD codes with good parameters including symplectic MDS codes are obtained.  
In particular, we also construct some binary symplectic LCD codes, which are equivalent to quaternary TrH ACD codes that 
outperform best-known quaternary Hermitian LCD codes.

Generally, it is difficult to determine the minimum symplectic distance or the symplectic weights distribution of a symplectic SO code or a symplectic LCD code. 
Although there are explicit minimum symplectic distances of symplectic SO and LCD codes obtained in this paper, 
it would be interesting to discuss symplectic weights distributions of these symplectic SO codes and symplectic LCD codes or 
to construct more symplectic SO codes and symplectic LCD codes with explicit minimum symplectic distances. 
In addition, it would also be interesting to apply symplectic SO and LCD codes obtained in this paper to construct QECCs and maximal entanglement EAQECCs.


\section*{Conflict of Interest} The authors affirm that no conflicts of interest exist.  

\section*{Data Availability} Data used for the research has been described in the article.

\end{sloppypar}

\section*{Appendix} 

Here, we give more details on permutation matrices used in Table 2 as follows:  
\begin{align*}
    P_{46}=(& 23\ 28\ 2\ 20\ 32\ 18\ 22\ 3\ 17\ 6\ 46\ 15\ 36\ 27\ 14\ 43\ 16\ 39\ 26\ 38\ \\ 
            & 12\ 42\ 4\ 10\ 40\ 31\ 37\ 24\ 7\ 29\ 5\ 19\ 44\ 1\ 41\ 13\ 11\ 45\ 21\ 35\ \\ 
            & 30\ 9\ 25\ 33\ 34\ 8), 
\end{align*}
\begin{align*}
    P_{52}=(& 22\ 29\ 50\ 46\ 35\ 37\ 42\ 45\ 18\ 30\ 36\ 2\ 25\ 43\ 10\ 5\ 26\ 11\ 24\ 27\ \\
    & 20\ 49\ 17\ 23\ 41\ 47\ 7\ 6\ 51\ 48\ 12\ 16\ 34\ 31\ 44\ 15\ 3\ 13\ 52\ 19\ \\ 
    & 14\ 33\ 38\ 40\ 9\ 8\ 32\ 39\ 4\ 21\ 28\ 1), 
\end{align*}
\begin{align*}
    P_{56}=(& 42\ 15\ 52\ 24\ 38\ 11\ 12\ 18\ 6\ 39\ 51\ 53\ 2\ 30\ 46\ 21\ 29\ 3\ 13\ 19\ \\ 
           & 49\ 36\ 48\ 9\ 7\ 31\ 41\ 50\ 40\ 8\ 4\ 25\ 1\ 47\ 34\ 5\ 33\ 45\ 10\ 27\ 14\ \\
           & 28\ 16\ 23\ 26\ 37\ 54\ 43\ 35\ 32\ 55\ 56\ 44\ 22\ 20\ 17), 
\end{align*}
\begin{align*}
    P_{58}=(& 11\ 38\ 18\ 51\ 42\ 52\ 19\ 48\ 56\ 55\ 22\ 20\ 7\ 5\ 4\ 21\ 34\ 14\ 40\ 35\ \\ 
            & 27\ 24\ 41\ 44\ 9\ 10\ 23\ 6\ 8\ 31\ 46\ 39\ 57\ 47\ 36\ 45\ 54\ 3\ 50\ 33\ \\ 
            & 49\ 1\ 15\ 58\ 29\ 16\ 26\ 2\ 13\ 25\ 17\ 28\ 12\ 53\ 43\ 37\ 32\ 30), 
\end{align*}
\begin{align*}
     P_{62}=(& 1\ 23\ 42\ 44\ 62\ 3\ 25\ 27\ 36\ 48\ 47\ 29\ 11\ 34\ 20\ 16\ 30\ 18\ 51\ 31\ 4 \\ 
             & 52\ 33\ 7\ 57\ 19\ 49\ 58\ 2\ 43\ 8\ 26\ 5\ 37\ 59\ 40\ 6\ 55\ 61\ 9\ 46\ 15\ 21\ \\ 
             & 22\ 38\ 32\ 13\ 28\ 53\ 12\ 17\ 24\ 56\ 10\ 60\ 39\ 50\ 35\ 14\ 41\ 45\ 54), 
\end{align*}
\begin{align*}
     P_{63}=(& 53\ 29\ 63\ 14\ 44\ 47\ 46\ 4\ 51\ 59\ 11\ 20\ 10\ 23\ 13\ 37\ 42\ 9\ 26\ 34\ 12 \\ 
             & 49\ 38\ 30\ 62\ 56\ 16\ 55\ 28\ 33\ 3\ 61\ 40\ 6\ 5\ 35\ 22\ 24\ 52\ 50\ 25\ 7\ \\ 
             & 18\ 39\ 36\ 31\ 8\ 21\ 27\ 57\ 17\ 60\ 41\ 58\ 19\ 43\ 54\ 48\ 1\ 32\ 15\ 45\ 2),  
\end{align*}
\begin{align*}
     P_{64}=(& 48\ 32\ 16\ 55\ 9\ 43\ 23\ 46\ 49\ 10\ 60\ 62\ 40\ 22\ 3\ 38\ 8\ 34\ 59\ 35\ 36 \ 30 \\ 
             & 1\ 6\ 63\ 57\ 5\ 44\ 15\ 33\ 53\ 4\ 25\ 31\ 29\ 45\ 58\ 17\ 64\ 11\ 14\ 24\ 52\ 42 \\ 
             & 18\ 20\ 51\ 37\ 56\ 61\ 41\ 13\ 19\ 7\ 50\ 2\ 12\ 21\ 26\ 54\ 39\ 27\ 47\ 28),  
\end{align*}
\begin{align*}
    P_{70}=(& 56\ 23\ 40\ 64\ 63\ 61\ 39\ 41\ 55\ 33\ 13\ 22\ 26\ 17\ 48\ 18\ 36\ 20\ 67\ 46 \\ 
            & 59\ 25\ 19\ 45\ 31\ 51\ 38\ 65\ 50\ 57\ 43\ 60\ 9\ 54\ 10\ 52\ 70\ 14\ 30\ 44\  \\ 
            & 62\ 66\ 58\ 5\ 35\ 3\ 42\ 6\ 7\ 21\ 16\ 47\ 53\ 28\ 8\ 49\ 29\ 34\ 4\ 27\ 15\ 37 \\ 
            & 69\ 11\ 32\ 1\ 68\ 12\ 2\ 24),  
\end{align*}
\begin{align*}
    P_{72}=(& 36\ 25\ 56\ 57\ 22\ 35\ 29\ 23\ 8\ 31\ 3\ 38\ 19\ 70\ 45\ 68\ 10\ 72\ 53\ 2\ 42 \\ 
            & 58\ 18\ 51\ 61\ 63\ 27\ 54\ 52\ 46\ 13\ 43\ 11\ 62\ 17\ 28\ 71\ 41\ 66\ 47\ 14\ \\ 
            & 6\ 12\ 48\ 26\ 65\ 67\ 30\ 40\ 24\ 7\ 33\ 32\ 34\ 64\ 55\ 4\ 50\ 49\ 69\ 37\ 60 \\ 
            & 15\ 21\ 16\ 1\ 44\ 20\ 39\ 9\ 59\ 5),  
\end{align*}
\begin{align*}
    P_{74}=(& 49\ 6\ 67\ 11\ 21\ 36\ 71\ 64\ 52\ 7\ 59\ 10\ 42\ 66\ 28\ 27\ 30\ 43\ 3\ 38\ \\ 
            & 20\ 14\ 16\ 40\ 60\ 15\ 5\ 41\ 54\ 73\ 17\ 19\ 26\ 57\ 62\ 32\ 2\ 53\ 61\ 34\ \\
            & 9\ 74\ 8\ 1\ 31\ 55\ 22\ 44\ 29\ 68\ 39\ 63\ 33\ 56\ 50\ 13\ 72\ 58\ 69\ 35\ \\ 
            & 70\ 51\ 24\ 47\ 4\ 12\ 45\ 48\ 46\ 18\ 25\ 37\ 23\ 65).   
\end{align*}

\end{document}